\documentclass[twocolumn,10pt]{IEEEtran}
\usepackage{amsfonts,amssymb,amsmath,calc,graphicx,subfigure}
\usepackage{psfrag,array,epsfig,pstricks,multicol,verbatim,moreverb,mathrsfs}

\newtheorem{theorem}{Theorem}
\newtheorem{lemma}{Lemma}
\newtheorem{corollary}{Corollary}

\def\1{\textbf{1}}

\begin{document}

\title{Convexity Conditions for 802.11 WLANs}
\author{Vijay G. Subramanian, Douglas J. Leith\\Hamilton Institute, NUI Maynooth\thanks{This material is based upon works supported by the HEA PRTLI Network Maths project and Science Foundation Ireland under Grant No. 07/IN.1/I901. }}
\maketitle

\begin{abstract}
In this paper we characterise the maximal convex subsets of the (non-convex) rate region in 802.11 WLANs.   In addition to being of intrinsic interest as a fundamental property of 802.11 WLANs,  this characterisation can be exploited to allow the wealth of convex optimisation approaches to be applied to 802.11 WLANs.  
\end{abstract}

\section{Introduction}


We establish a number of fundamental convexity properties of the rate region of 802.11 WLANs.   Firstly, we establish a simple constraint that determines the station transmission attempt probabilities on the rate region boundary.   Secondly, we show that, while the rate region is non-convex, its complement in the positive orthant is strictly convex.  This property is illustrated in Fig. \ref{fig:rateregion} where the shaded area indicates the rate region and the unshaded area above the rate region boundary indicates its complement in the positive orthant.   Thirdly, we obtain a complete and explicit characterisation of the maximal convex subsets of the rate region, two such subsets being indicated by the hashed areas in Fig. \ref{fig:rateregion}.     It is important to note here that  ``obvious'' constraints  do not yield convex subsets of the rate region, let alone maximal subsets.  Examples of constraints for which it is straightforward to show (proofs omitted due to lack of space) that convexity does not result (for WLANs with $n>2$ stations) include: 
\begin{enumerate}
\item Constraining the maximum transmission attempt probabilities, \emph{i.e.} enforcing $\tau_i\le\tau_{max,i}$ for station $i$;
\item Constraining the maximum value of $\sum_{i=1}^n\tau_i$ where $n$ is the number of stations in the WLAN;
\item Constraining the maximum value of the WLAN collision probability;
\item Constraining the minimum value of the WLAN idle probability $P_{idle}$.
\end{enumerate}

Our results complement the recent observation in \cite{logconvexity10} that the 802.11 rate region is log-convex.    As well as being of interest in their own right, our results provide the basis for applying powerful convex optimisation methods to the analysis and design of fair throughput allocations for 802.11 WLANs -- we discuss this in more detail in Section \ref{sec:opt} below.

The paper is organised as follows.   We first introduce our network model in Section  \ref{sec:nm}, then consider the rate region boundary in Section \ref{sec:rr}.  In Section \ref{sec:convexity} we establish that the complement of the rate region is strictly convex and in Section \ref{sec:maximal} characterise the maximal convex subsets of the rate region.   We summarise our conclusions in Section \ref{sec:concl}.

\section{Network Model}\label{sec:nm}

The 802.11e standard extends and subsumes the standard 802.11 DCF (Distributed Coordinated Function) contention mechanism by allowing the adjustment of MAC parameters that were previously fixed, including $DIFS$ (called $AIFS$ in 802.11e), $CW_{min}$ and $CW_{max}$.
 In addition, 802.11e adds a TXOP mechanism that specifies the duration during which a station can keep transmitting without releasing the channel once it wins a transmission opportunity. In order not to release the channel, a SIFS interval is inserted between each frame-ACK pair. A successful transmission round consists of multiple frames and ACKs. By adjusting this time, the number of framess that may be transmitted by a station at each transmission opportunity can be controlled.  A salient feature of
the TXOP operation is that, if a large TXOP is assigned and there are not enough
packets to be transmitted, the TXOP period is ended immediately to avoid wasting
bandwidth.

\begin{figure}
\centering
\includegraphics[width=0.49\columnwidth]{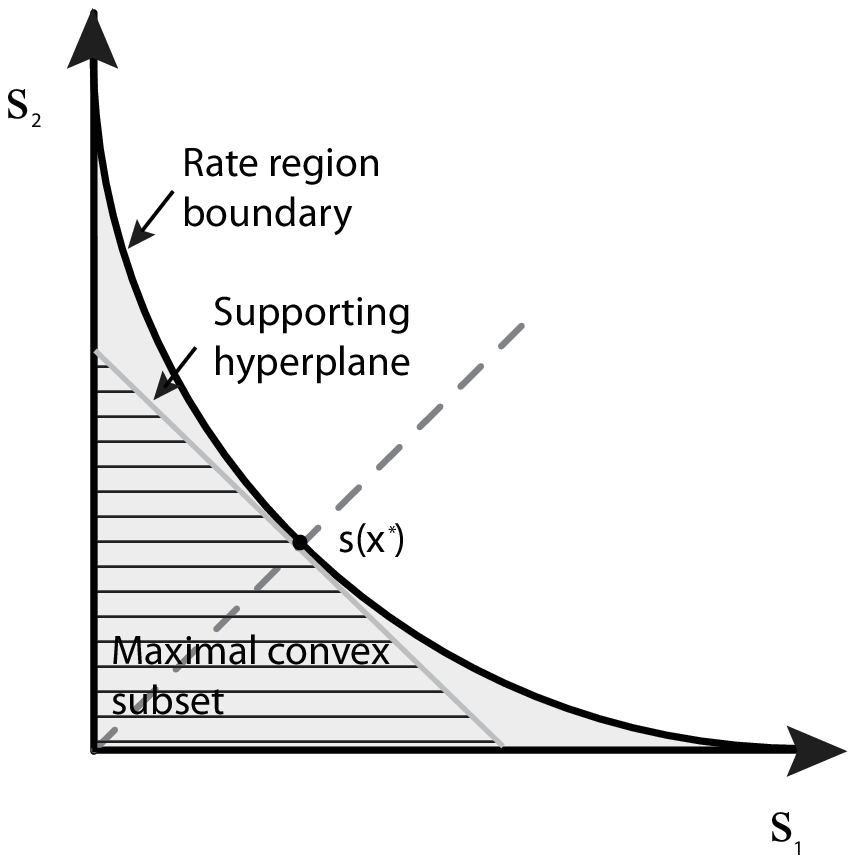}
\includegraphics[width=0.49\columnwidth]{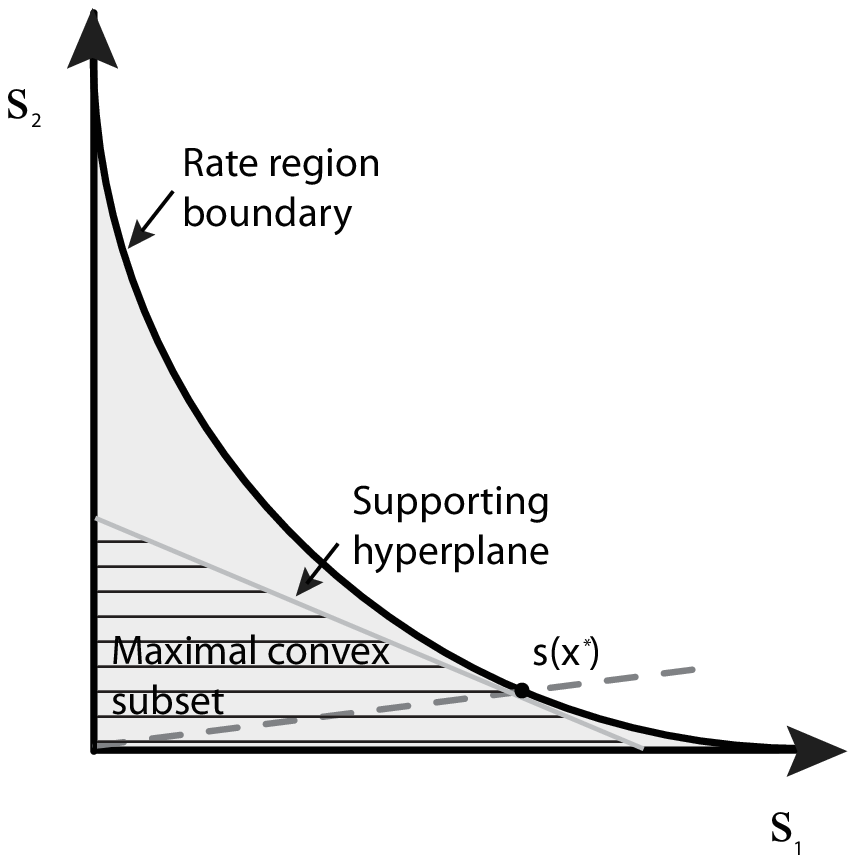}
\caption{Illustrating the rate region (shaded) of an 802.11 WLAN with two stations.  The complement of the rate region in the positive orthant (unshaded) is strictly convex.  Also shown are the maximal convex subsets corresponding to two different boundary points.}\label{fig:rateregion}
\end{figure}

We consider an 802.11e WLAN with $n$ stations.   As described in \cite{David_TON_2007,clifford06}, we divide time into MAC slots, where each MAC slot may consist either of a PHY idle slot, a successful transmission or a colliding transmission (where more than one station attempts to transmit simultaneously).  Let $\tau_i$ denote the probability that station $i$ attempts a transmission.  The mean throughput of station $i$ is
\begin{align*}
s_i (\tau) = \frac{P_{succ,i} D_i}{\sigma P_{idle} + \sum_{i=1}^nT_{s,i} P_{succ,i}+T_c(1-P_{idle}-P_{succ})} 
\end{align*}
where $P_{idle}=\prod_{k=1}^n (1-\tau_k)$, $P_{succ,i}=\tau_i \prod_{k=1,k\ne i}^n(1-\tau_k)$, $P_{succ}=\sum_{i=1}^nP_{succ,i}$, $\tau=[\tau_1\ ...\ \tau_n]^T$, $D_i$ is the mean number of bits sent by station $i$ in a successful transmission, $\sigma$ is the PHY idle slot duration, $T_{s,i}$ is the mean duration of a successful transmission by station $i$ and $T_c$ the mean duration of a collision.    Note that $T_{s,i}$ and $D_i$ are allowed to depend on the station to encompass situations where stations may transmit different sized TXOP bursts on winning a transmission opportunity.

We will assume that frame transmissions are of duration $T_c$, in which case the collision duration $T_c$ is invariant with the station attempt probabilities $\tau_i$ (if stations use frames of different duration then the duration of a collision would depend on the specific set of stations involved in a collision and so on the attempt probabilities).    A TXOP burst by station $i$  consists of a sequence of $N_i=T_{s,i}/T_c$ frame transmissions each of duration $T_c$ and  $L_i=D_i/N_i$ is the size, in bits, of the payload of each frame. We suppose that there are upper/lower bounds on the admissible TXOP burst size,  \emph{i.e.} $1\le \underbar{N}_i\le N_i\le \bar{N}_i$.

It will prove useful to work in terms of the quantity $x_i=\tau_i/(1-\tau_i)$ rather than $\tau_i$ -- observe that $x_i\in[0,\infty)$ for $\tau_i\in[0,1)$.  With this transformation we have that $P_{idle}=1/\prod_{k=1}^n (1+x_k)$ and $P_{succ,i}=x_iP_{idle}$ and
\begin{align}
s_i(x,N)=\frac{N_ix_i}{X(x,N)}\frac{L_i}{T_{c}}
\label{eq:tput_exp}
\end{align}
where 
\begin{align}\label{eq:Xdef}
X(x,N)=a+\sum_{k=1}^n (N_k-1)x_k + \prod_{k=1}^n (1+x_k)-1
\end{align}
with $a=\sigma/T_c$. For a fixed $\sigma$, the throughputs scale with $T_c$ so henceforth we assume that $T_c=1$.


\section{Rate Region}\label{sec:rr}
The rate region is the set $R$ of achievable throughput vectors $s(x,N)=[s_1\ ...\ s_n]^T$ as the vector $x$ of attempt rates ranges over domain $[0,\infty)^n$ and the vector $N$ of TXOP burst sizes range over $\prod_{k=1}^n[\underbar{N}_k,\bar{N}_k]$.    
In this section we establish some basic properties of  the boundary of the rate region.   

\subsection{Rate Region Boundary}

\begin{lemma}\label{lem:one}The boundary of the rate region is the set of throughput vectors $s(x^*,\bar{N})$ with $x^*\in B$, where 
\begin{align}\label{eq:boundary}
B=\left\{x:h(x)=1\right\}
\end{align}
and 
\begin{align}\label{eq:h}
h(x)= \sum_{i=1}^n \frac{x_i}{1+x_i}+\frac{1-a}{\prod_{j=1}^n(1+x_j)}
\end{align}
\end{lemma}
\begin{proof}
Take a vector $y$, with $y_i > 0$ normalised such that $\sum_i y_i=1$, and set $x_i = \lambda \bar{x}_i$, where $\lambda\ge 0$ and $\bar{x}_i=y_i/(L_i N_i)$.  From (\ref{eq:tput_exp}), the vector of station throughputs is then $s=\lambda/X(\lambda \bar{x},N) y$.  Since $\lambda$, $X(\bullet,N)$ are scalars it can be seen that varying $\lambda$ adjusts the position of the throughput vector on the ray in direction $y$ passing through the origin. To determine the rate region boundary we need to find the value of $\lambda$ that maximises  $\lambda/X(\lambda \bar{x},N)$.  
By inspection of the first derivative it can be verified that $\lambda/X$ is monotonically increasing in the TXOP burst sizes $N_i$. Therefore at the rate region boundary we must have $N_i =\bar{N}_i, i=1,...,n$.
It can also be verified by inspection of the second derivative that $\lambda/X$ is a concave function of $\lambda$ and so has a unique turning point.  To determine the turning point of $ \lambda/X$, differentiating $\lambda/X$ with respect to $\lambda$ yields
$$
\frac{X - \lambda  \bigg(\sum_{i=1}^n \frac{y_i(N_i-1)}{L_i} + \sum_{i=1}^n \frac{y_i}{L_iN_i}\prod_{j\ne i} \Big(1+\frac{\lambda y_i}{L_iN_i}\Big)\bigg)}{X^2}
$$
 and setting this derivative equal to zero we have that the $\lambda^*$ corresponding to the turning point is the unique positive root of
 $$
 \sum_{i=1}^n \frac{\lambda^* y_i}{L_iN_i}\prod_{j\ne i} \Big(1+\frac{\lambda^* y_i}{L_iN_i}\Big)+1-a = \prod_{i=1}^n\Big(1+\frac{\lambda^* y_i}{L_iN_i}\Big)
$$
Substituting, we therefore have that the turning point (\emph{i.e.}, the boundary of the rate-region) satisfies
\begin{align*}
\sum_{i=1}^n \frac{x_i^*}{1+x_i^*}+\frac{1-a}{\prod_{j=1}^n(1+x_j^*)}=1
\end{align*}
where $x^*=\lambda^* \bar{x}$, as stated in the lemma.  Note that this boundary condition can also be rewritten in terms of $\tau$ as 
$\sum_{i=1}^n\tau^*_i +(1-a)\prod_{i=1}^n(1-\tau^*_i)=1$, i.e., $\sum_{i=1}^n\tau^*_i + (1-a) P_{idle} =1$. 
\end{proof}

This lemma generalises the result for ALOHA networks of \cite{MasseyMathys,Post} (which is for the specific situation where $a=1$, $N_i\equiv 1$ and $L_i\equiv 1$).  Observe that the TXOP burst sizes $N_i$ do not play a role in determining the boundary $x^*$, although they will influence the value of the throughput vector $s(x^*,\bar{N})$.

\subsection{Tangent Hyperplanes to Rate Region Boundary}\label{sec:thp}

\begin{lemma}
The tangent hyperplane to point $s(x^*,\bar{N})$ on the rate region boundary is the set
\begin{align}
T(x^*)=\left\{s: \sum_{i=1}^n b_i(x^*)s_i = \frac{1}{\prod_{j=1}^n(1+x_j^*)}\right\}
\end{align}
where
\begin{align}\label{eq:bi}
b_i(x^*) = \frac{1}{L_i\bar{N}_i}\left(\frac{\bar{N}_i-1}{\prod_{j=1}^n (1+x_j^*)} + \frac{1}{1+x_i^*}\right)
\end{align}
\end{lemma}
\begin{proof}
Taking the derivative of the station throughput with respect to the $x_i$, from (\ref{eq:tput_exp}) we have
\small
\begin{align*}
\frac{\partial s_i(x)}{\partial x_k} = 
\begin{cases}
\frac{L_i\bar{N}_i}{X^2}\left( X-(\bar{N}_i-1)x_i-\frac{x_i}{1+x_i}\prod_{j=1}^n(1+x_j)\right) & k=i \\
- \frac{L_i\bar{N}_i}{X^2}\left( (\bar{N}_k-1)x_i+\frac{x_i}{1+x_k}\prod_{j=1}^n(1+x_j)\right) & k\ne i
\end{cases}
\end{align*}
\normalsize
The normal vector $b(x^*)$ to the tangent hyperplane at point $s(x^*,\bar{N})$ on the rate region boundary solves $\sum_{i=1}^n b_i(x^*) \partial s_i(x^*,\bar{N})/\partial x_k = 0$ $\forall k=1,...,n$.   Making use of Lemma \ref{lem:one} characterising boundary points, it can be verified that the vector $b(x^*)$ stated in the lemma  is one such normal vector. 
\end{proof}

\section{Convexity of Nonachievable Region}\label{sec:convexity}

Let $\bar{R}$ denote the complement of the rate region $R$ in the positive orthant.  $\bar{R}$ is the set of nonachievable throughput vectors lying \emph{outside} the rate region $R$, and is given by
$$
\bar{R}=\{u: u=\lambda s(x^*,\bar{N}), \lambda>1, x^*\in B\}
$$
We now show that while the rate region $R$ is non-convex, $\bar{R}$ is strictly convex.
\begin{lemma}\label{eq:convexity}The set $\bar{R}$ is strictly convex.
\end{lemma}

Before proving Lemma \ref{eq:convexity} we note the following fact in \cite{Post} using Bessel's inequality.
\begin{lemma}\label{lem:post}\cite{Post} 
 Let vector $r\in\mathbb{R}^n$ be such that $0\leq r_j \leq 1$ for $j=1,\dotsc,n$ and $\sum_{j=1}^nr_j = n-1$.   Let vector $z\in\mathbb{R}^n$ satisfy $r^Tz=0$.  Then $\sum_{j=1}^n\sum_{i=1}^{j-1} z_iz_j < 0$. 
\end{lemma}

\begin{proof}\emph{Lemma \ref{eq:convexity}}
A supporting hyperplane of set $\bar{R}$ at boundary point $x^*$ is such that (i) set $\bar{R}$ is entirely contained in one closed half-space and (ii) $x^*$ lies on the hyperplane.  By the Tietze-Nakajima Theorem \cite{Tietze}, the open and connected set $\bar{R}$ is convex if for every boundary point $x^*$ there is a locally supporting hyperplane.   This is satisfied if for all $x^*\in B$ and $y^*\in B$ sufficiently close to $x^*$, $s(y^*,\bar{N})$ lies above the tangent plane $T(x^*)$.  That is, it is sufficient to show that
\begin{align}\label{eq:convexcond}
\sum_{i=1}^n b_i(x^*)s_i(y^*,\bar{N}) > \frac{1}{\prod_{j=1}^n(1+x_j^*)}
\end{align}
for all $y^*$ sufficiently close to $x^*$.  Substituting for $b_i(x^*)$ from (\ref{eq:bi}), the LHS can be rewritten as
\small
\begin{align*}
\frac{1}{\prod_{j=1}^n(1+x_j^*)} \left(\frac{\sum_{i=1}^n(\bar{N}_i-1)y_i^*+\prod_{j=1}^n(1+x^*_j)\sum_{i=1}^n\frac{y_i^*}{1+x_i^*}}{X(y^*,\bar{N})}\right)
\end{align*}
\normalsize
Condition (\ref{eq:convexcond}) is satisfied if the term in brackets is greater than unity, \emph{i.e.} provided 
$$
\prod_{j=1}^n(1+x^*_j)\sum_{i=1}^n\frac{y^*_i}{1+x_i^*}>\prod_{j=1}^n(1+y^*_j)+a-1
$$
Letting $y_i^*=x_i^*+\delta_i$ and using Lemma \ref{lem:one} for points $x^*$ and $y^*$, this condition becomes
\begin{align}\label{eq:almostthere}
\prod_{j=1}^n(1+x^*_j)\left(1+\sum_{i=1}^n\frac{\delta_i}{1+x_i^*}\right)>\prod_{j=1}^n(1+x^*_j+\delta_j)
\end{align}
Expand the RHS as
$$
\prod_{j=1}^n(1+x^*_j)\left(1+\sum_{i=1}^n\frac{\delta_i}{1+x^*_i}+ \sum_{j=1}^n\sum_{i=1}^{j-1}\frac{\delta_i}{(1+x^*_i)}\frac{\delta_j}{(1+x^*_j)} +\epsilon\right)
$$
where $\epsilon$ involves cubic and higher terms.   Condition (\ref{eq:almostthere}) then can be rewritten as
\begin{align*}
\sum_{j=1}^n\sum_{i=1}^{j-1}\frac{\delta_i}{(1+x^*_i)}\frac{\delta_j}{(1+x^*_j)} +\xi < 0
\end{align*}
for some $\xi$ that involves cubic and higher terms.
For $\delta$ sufficiently small, the sign of the LHS is determined by the quadratic term. 

Since $y^*\in B$, from the first-order conditions on points on the boundary we get that the perturbation $\delta$ needs to be orthogonal to $\nabla h(x) |_{x=x^*}$, \emph{i.e.} $\sum_{i=1}^n \delta_i \partial h(x) / \partial x_i |_{x=x^*} = 0$ where $h(x)$ is given by (\ref{eq:h}). Now we have
\begin{align*}
\sum_{i=1}^n \delta_i \frac{\partial h(x)}{\partial x_i} \Big|_{x=x^*} & = \sum_{j=1}^n \frac{\delta_j}{1+x_j^*} \left(\frac{1}{1+x_j^*} - \frac{1-a}{\prod_{i\in\mathcal{M}} (1+x_i^*)} \right) \\
& = \sum_{j=1}^n \frac{\delta_j}{1+x_j^*} \left( \sum_{i=1,i\neq j}^n \frac{x_i^*}{1+x_i^*} \right)
\end{align*}
where we use the boundary property of $x^*$. Now define the following
\begin{align*}
z_j := \frac{\delta_j}{1+x_j^*},\; r_j :=\frac{\sum_{i=1,i\neq j}^n \frac{x_i^*}{1+x_i^*}}{1-\frac{1-a}{\prod_{i=1}^n (1+x_i^*)}} \; \forall\; j=1,\dotsc,n 
\end{align*}
It can be verified using the boundary property of $x^*$ that $0\leq r_j \leq 1$ for $j=1,\dotsc,n$ and $\sum_{j=1}^nr_j = n-1$ with $z$ and $r$ orthogonal.   Then by Lemma \ref{lem:post}, $ \sum_{j=1}^n\sum_{i=1}^{j-1} z_iz_j < 0$ as required. 
\end{proof}

\section{Maximal Convex Subsets of the Rate Region}\label{sec:maximal}

We are now in a position to state our main result.
\begin{theorem}\label{th:main} For any point $s(x^*)$ on the boundary of rate region $R$, the set 
\small
\begin{align*}
C(x^*)=\left\{s:\sum_{i=1}^n \alpha_i(x^*)s_i \le 1, s_i\ge 0\right\}
\end{align*}
\normalsize
is the maximal convex subset of $R$ containing $s(x^*)$, where $\alpha_i(x^*)=\frac{\bar{N}_i-1 + \prod_{j=1,j\ne i}^n(1+x_j^*)}{L_i \bar{N}_i}$.
\end{theorem}
\begin{proof}
For any point $s(X^*)$ on the boundary of rate region $R$, the tangent hyperplane is also a supporting hyperplane of the nonachievable set $\bar{R}$ (due to the convexity of $\bar{R}$, Lemma \ref{eq:convexity}).   It follows that the set $C(x^*)$ of points lying below the hyperplane at $s(x^*)$ lie within rate region $R$ \emph{i.e.} $C(x^*)\subset R$.   The set $C(x^*)$ is given by
\begin{align}\label{eq:maxconvex}
C(x^*)=\left\{s:\sum_{i=1}^n b_i(x^*)s_i \le \frac{1}{\prod_{j=1}^n(1+x_j^*)}, s_i\ge 0\right\}
\end{align}
Substituting from (\ref{eq:bi}) for the $b_i(x^*)$ yields the expression for $C(x^*)$ in the statement of the theorem. Since $C(x^*)$ is formed by the intersection of $\mathbb{R}_+^n$ and the (unique) supporting half-space to the union of $\bar{R}$ and the boundary of $R$, maximality follows.   
\end{proof}

This theorem in illustrated in Fig. \ref{fig:rateregion} for a WLAN where $\bar{N}_i=\bar{N}$, $i=1,...,n$.   In Fig. \ref{fig:rateregion}(a) the boundary point is the symmetric one where $s_i(x^*)=s_j(x^*)$, $\forall i,j=1,...,n$.   The supporting hyperplane is indicated by the 45$^o$ line and the hashed area indicates the maximal convex subset of $R$ containing $s(x^*)$.   Fig. \ref{fig:rateregion}(b)  shows an asymmetric example where $s_i(x^*)\ne s_j(x^*)$ when $i\ne j$.    We summarise the symmetric case in the following corollary to Theorem \ref{th:main}:
\begin{corollary}
When $\bar{N}_i=\bar{N}$, $i=1,...,n$, the maximal convex subset associated with the symmetric boundary point $s(x^*)$ \emph{i.e.} $s_i(x^*)=s_j(x^*)=s^*$, $i,j=1,...,n$ is
\begin{align*}
\{s:\sum_{i=1}^n s_i \le 1/\alpha^*, s_i\ge 0\}
\end{align*}
where $\alpha^*=\frac{\bar{N} -1+ (1+s^*)^{n-1}}{X(x^*)}$.
\end{corollary}

\section{Application to Utility-based Optimisation}\label{sec:opt}

%

In this section we briefly illustrate use of the results of the previous sections in the design of fair throughput allocations in mesh networks of 802.11 WLANs. In previous work \cite{logconvexity10,maxmin10} we established the log-convexity of the rate-region of mesh networks of 802.11 WLANs.   By working in terms of log-transformed rates, this then implies that one can use the Network-Utility-Maximization (NUM) framework of Kelly~\cite{Kelly97} for the class of utility functions $U(\cdot)$ such that $U(\exp(\cdot))$ is concave.  The commonly used iso-elastic/constant relative risk aversion family of utility functions~\cite{Arrow1970,Pratt1964,Merton1971,Venter1983,Xie2000,MoWalrand2000} given by
\begin{align*}
U(x) = 
\begin{cases}
\frac{x^{1-\alpha}-1}{1-\alpha} & x\geq 0,\; \alpha\geq 0, \; \alpha \neq 1 \\
\log(x) & x\geq 0, \; \alpha=1
\end{cases}
\end{align*}
satisfy the above requirement only when $\alpha \geq 1$. The requirement that $U(\exp(\cdot))$ is concave when $U(\cdot)$ is already concave is a stringent one and excludes from consideration many families of utility functions used by economists to model user behaviours~\cite{Xie2000}.   For example,
\begin{enumerate}
\item[(i)] Hyperbolic absolute risk aversion family of utility functions~\cite{Merton1971}, which is given by
\begin{align*}
U(x) = \frac{\alpha}{1-\alpha}\left[ \left(\beta+\frac{x}{\gamma}\right)^{1-\alpha}-1\right], \; \beta+\frac{x}{\gamma} > 0
\end{align*}
for $\alpha\neq 0$ and $\alpha\neq 1$ (via limits for $\alpha\in\{0, 1, +\infty\}$);
\item[(ii)] Linear exponential family of utility functions~\cite{Frank1990}, given by
\begin{align*}
U(x) = x - \beta \exp(-\alpha x), \; x\geq 0 
\end{align*}
where $\beta, \alpha \geq 0$; and
\item[(iii)] Power risk aversion family of utility functions~\cite{Venter1983,Xie2000}, which is based on the Weibull distribution and is given by
\begin{align*}
U(x) = \frac{1}{\beta}\left[1-\exp\left( -\beta \left(\frac{x^{1-\alpha}-1}{1-\alpha} \right)\right) \right], \; \beta, \alpha \geq 0,
\end{align*}
for $x\geq 0$ where the edge cases of $\beta=0$ and $\alpha=1$ are defined via limits.
\end{enumerate}
For suitable parameter settings these families have the decreasing absolute risk aversion property~\cite{Arrow1970,Pratt1964}, they have $d^3U(x)/dx^3 > 0$ and are increasing and concave, and therefore fit empirically observed user behaviours~\cite{Xie2000}.   Yet for most parameter settings $U(\exp(\cdot))$ is not a concave function for these families, and thus utility optimisation cannot be addressed via the approach in \cite{logconvexity10,maxmin10}.  Since the original motivation for the NUM framework~\cite{Kelly97} was to bring in economic considerations to rate allocation in networks, this potentially represents a major deficiency that can addressed using the results in the present paper. With this in mind, we can revisit the setting in \cite{maxmin10} of a mesh network formed from 802.11 WLANs/cliques created using appropriate frequency assignment.   Using Theorem~\ref{th:main} we can work in an appropriate convex subset of the rate-region of the network. For this we assume that an appropriate operating point on the boundary is chosen for each WLAN/clique,  the subset rate-region is then an intersection of polytopes given by \eqref{eq:maxconvex} and the theory from \cite{Kelly97} directly applies.

We illustrate this using a specific network with three flows shown in Figure~\ref{fig:example} and using two utility functions from the power risk aversion family~\cite{Xie2000}, namely, $\alpha=0.1$ and $\beta=1$ giving $U_1(x)=1-\exp\left(-\left(x^{0.9}-1\right)/0.9\right)$, and $\alpha=2.0$ and $\beta=1$ giving $U_2(x)=1-\exp\left(\left(x^{-1}-1\right)\right)$. We will compare this with the $\log(\cdot)$ utility function corresponding to proportional fairness. In all cliques we assume that $a=1/9$ and that the TXOP value is set to 1, i.e., its minimum possible value. For flow 1 we assume that $L/T_c$ corresponds to 12 Mbps wherever the flow is active; the corresponding numbers for flow 2 and flow 3 are assumed to be 6 Mbps and 12 Mbps, respectively. From symmetry of the problem it is clear that one can find the various utility optimal solutions by choosing the same operating point for flow 2 in cliques 2 and 3; let this be $x_2^*$. Then we have that $x_1^*=x_3^*=a/x_2^*$. The optima that result are as follows: for proportional fairness $x_2^*=0.2094$; for utility 1 $x_2^*=0.3767$; and for utility 2 $x_2^*=0.3516$. Given these operating points for the different cliques, one can use Theorem~\ref{th:main} to then find the optimal solution using standard convex optimization techniques. The boundary of clique 2, the respective optimizers and the maximal convex subsets from Theorem~\ref{th:main} for this problem are illustrated in Figure~\ref{fig:output}. Note that only the proportionally fair optimizer can be determined using the log-convexity ideas in \cite{logconvexity10,maxmin10}. In this simple example one can directly calculate the optimizers but the same idea carries through to the more general topologies presented in \cite{maxmin10}.

\begin{figure}
\centering
\includegraphics[width=0.99\columnwidth]{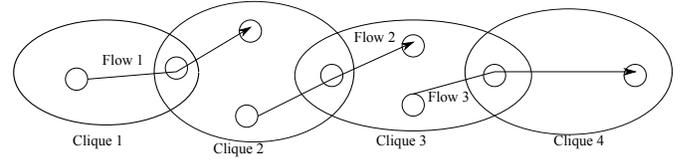}
\caption{Topology for numerical example with four cliques carrying three flows.}\label{fig:example}
\end{figure}

\begin{figure}
\centering
\includegraphics[width=0.99\columnwidth]{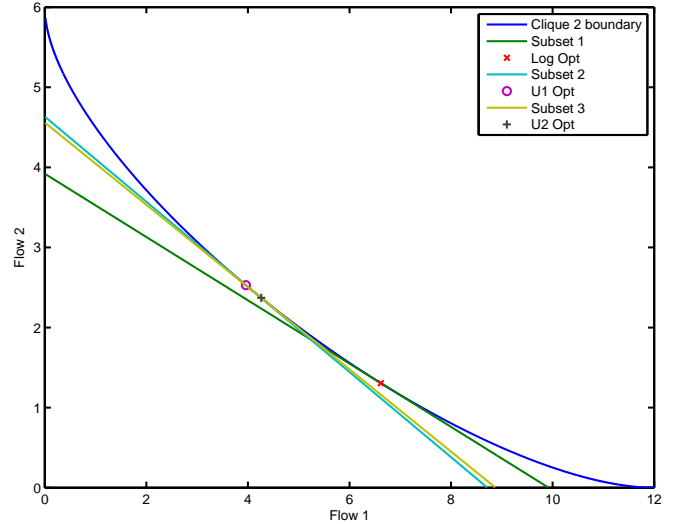}
\caption{Rate-region of clique 2 illustrated with the different utility optimal points along with the corresponding maximal convex subset.}\label{fig:output}
\end{figure}

\section{Conclusions}\label{sec:concl}

In this paper we characterise the maximal convex subsets of the (non-convex) rate region in 802.11 WLANs.   In addition to being of intrinsic interest as a fundamental property of 802.11 WLANs,  this characterisation can be exploited to allow the wealth of convex optimisation approaches to be directly applied to 802.11 WLANs.  In particular, standard utility-based fairness approaches can be applied for the important class of utility functions where $U(\exp(\cdot))$ is not concave.

{}

\end{document}